\documentclass[conference,a4paper]{IEEEtran}

\usepackage[english]{babel}
\usepackage[T1]{fontenc}
\usepackage[utf8]{inputenc}
\usepackage{subfig}
\usepackage{multirow}
\usepackage[font = scriptsize, labelfont = bf, labelsep = period, tableposition = top, singlelinecheck = false, justification = centering]{caption}
\usepackage{cite}
\usepackage{placeins}
\usepackage{amssymb}
\usepackage{amsmath}
\usepackage{amsthm}
\usepackage{latexsym}
\usepackage{wrapfig}
\usepackage{float}
\usepackage{mathtools}
\usepackage{bm}
\usepackage{listings}
\usepackage{textcomp}
\usepackage{setspace}
\usepackage[dvipsnames]{xcolor}
\usepackage{etoolbox}
\usepackage{hyperref}

\usepackage{algorithm,algorithmic}

\newtheorem{theo}{Theorem}[section]

\newtheorem{lem}[theo]{Lemma}
\theoremstyle{definition}

\newtheorem{rem}[theo]{Remark}

\newcommand{\U}{{\mathbb U}}
\newcommand{\F}{{\mathbb F}}

\newcommand{\C}{{\mathbb C}}

\newcommand{\cC}{{\mathcal C}}

\newcommand{\cG}{{\mathcal G}}
\newcommand{\cH}{{\mathcal H}}

\newcommand{\cP}{{\mathcal P}}

\newcommand{\cI}{{\mathcal I}}

\renewcommand{\cH}{{\mathcal H}}
\newcommand{\cS}{{\mathcal S}}
\newcommand{\cN}{{\mathcal N}}

\newcommand{\cW}{{\mathcal W}}

\newcommand{\bP}{{\mathbf P}}
\newcommand{\bH}{{\mathbf H}}
\newcommand{\bZ}{{\mathbf Z}}

\newcommand{\wcI}{\widetilde{\cI}}
\newcommand{\HI}{\bH_{\cI}}
\newcommand{\PI}{\bP_{\scriptsize\cI}}
\newcommand{\gf}{\mbox{$\bG_{\scriptsize{\bF}}$}}
\newcommand{\wb}{\bw_{\scriptsize{{\bf b}}}}
\newcommand{\HW}{\cH\cW_{\!N}}
\newcommand{\Fm}{\F_2^{2m}}
\newcommand{\wHI}{\widetilde{\bH_{\cI}}}

\newcommand{\sbt}{\raisebox{.2ex}{\mbox{$\scriptscriptstyle\bullet\,$}}}

\newcommand{\Sp}{\mathrm{Sp}}

\newcommand{\rk}{\mbox{${\rm rank\,}$}}

\newcommand{\cl}{\mbox{${\rm Cliff\!}$}}

\newcommand{\cs}{\mbox{\rm cs}\,}

\newcommand{\wt}{{\rm wt}}

\newcommand{\T}{\mbox{$^{\sf T}$}}

\newcommand{\inners}[2]{\mbox{$\langle{\,{#1}\,}|\,{#2}\,\rangle_{\rm s}$}}

\newcommand{\ov}[1]{\mbox{$\overline{#1}$}}
\newcommand{\FDP}{\bF_D(\bP)}
\newcommand{\FUS}{\bF_U(\bS)}
\newcommand{\GL}{\mathrm{GL}}
\newcommand{\Sym}{\mathrm{Sym}}

\newcommand{\diag}{\textup{diag}}

\newcommand{\bb}[1]{\mbox{\rm $\textbf{#1}$}}

\newcommand{\Imr}{\bI_{m|r}}
\newcommand{\Imrr}{\bI_{m|-r}}

\newcommand{\bA}{{\mathbf A}}

\newcommand{\bD}{{\mathbf D}}
\newcommand{\bE}{{\mathbf E}}
\newcommand{\bF}{{\mathbf F}}
\newcommand{\bG}{{\mathbf G}}
\newcommand{\bI}{{\mathbf I}}

\newcommand{\bS}{{\mathbf S}}
\newcommand{\ba}{{\mathbf a}}
\newcommand{\bc}{{\mathbf c}}
\newcommand{\bd}{{\mathbf d}}
\newcommand{\be}{{\mathbf e}}

\newcommand{\bs}{{\mathbf s}}

\newcommand{\bu}{{\mathbf u}}
\newcommand{\bv}{{\mathbf v}}
\newcommand{\bw}{{\mathbf w}}
\newcommand{\bx}{{\mathbf x}}
\newcommand{\by}{{\mathbf y}}
\newcommand{\bz}{{\mathbf z}}
\newcommand{\twomat}[2]{\mbox{$\left[\!\!\begin{array}{cc}{#1}&{#2}\end{array}\!\!\right]$}}
\newcommand{\fourmat}[4]{\mbox{$\left[\!\!\begin{array}{cc}{#1}&{#2}\\{#3}&{#4}\end{array}\!\!\right]$}}
\newcommand{\twomatv}[2]{\mbox{$\left[\!\!\begin{array}{c}{#1}\\{#2}\end{array}\!\!\right]$}}

\newcounter{alp}
\newcounter{ara}
\newcounter{rom}

\newif\ifcomment
\commenttrue

\begin{document}
\title{Reconstruction of Multi-user Binary Subspace Chirps}

\author{
  \IEEEauthorblockN{Tefjol Pllaha\IEEEauthorrefmark{1}, Olav Tirkkonen\IEEEauthorrefmark{1}, Robert Calderbank\IEEEauthorrefmark{2}}
	
 \IEEEauthorblockA{\small \IEEEauthorrefmark{1}Aalto University, Finland, e-mails: \{tefjol.pllaha, olav.tirkkonen\}@aalto.fi
    }
 \IEEEauthorblockA{\small \IEEEauthorrefmark{2}Duke University, NC, USA, e-mail: robert.calderbank@duke.edu
    }
}

\maketitle

\begin{abstract} 
We consider codebooks of Complex Grassmannian Lines consisting of Binary Subspace Chirps (BSSCs) in $N = 2^m$ dimensions. BSSCs are generalizations of Binary Chirps (BCs), their entries are either fourth-roots of unity, or zero. BSSCs consist of a BC in a non-zero subspace, described by an on-off pattern. Exploring the underlying binary symplectic geometry, we provide a unified framework for BSSC reconstruction---both on-off pattern and BC identification are related to stabilizer states of the underlying Heisenberg-Weyl algebra. In a multi-user random access scenario we show feasibility of reliable reconstruction of multiple simultaneously transmitted BSSCs with low complexity.  
\end{abstract}




\section{Introduction}
Codebooks of complex projective (Grassmann) lines, or tight frames,
have applications in multiple problems of interest for communications
and information processing, such as code division multiple access
sequence design~\cite{Viswanath1999}, precoding for multi-antenna
transmissions~\cite{Love2003} and network coding~\cite{Kotter2008}.
Contemporary interest in such codes arise, e.g., from deterministic
compressed sensing~\cite{DeVore2007,HCS08,Li2014,Wang2018,TC18}, virtual
full-duplex communication~\cite{GZ10}, mmWave
communication~\cite{Tsai2018}, and random
access~\cite{Calderbank2019}. In this paper, the main motivation will
come from a random access scenario, in particular from a {\it Massive
  Machine Type Communication} (MTC) scenario~\cite{Osseiran2016},
where the number of potentially accessing users may be extremely high, while
a majority of devices may be stationary. In such scenarios, encoding
and decoding complexity is of particular interest. To limit complexity
and power consumption for MTC devices, it is important that a limited
alphabet with small power variation is applied for transmission. 
Low decoding complexity is important for receiver implementation;
complexity should not grow as a function of the number of codewords.

Codebooks of Binary Chirps (BCs)~\cite{HCS08} provide an algebraically
determined set of Grassmannian line codebooks in $N=2^m$ dimensions,
with desirable properties; all entries are fourth root of unity and the
minimum distance is $1/\sqrt{2}$. The number of codewords is
reasonably large, growing as $2^{m(m+3)/2}$, while single-user
decoding complexity is ${\cal O}(N\log^2 N)$. Recently in~\cite{TC19}, we
expanded the set of Binary Chirps to Binary Subspace Chirps (BSSCs).
Taking the underlying binary symplectic geometry fully into account,
complex Grassmannian line codebooks are created with entries being either
scaled fourth-roots of unity, or zero. Comparing to BCs, the minimum
distance remains $1/\sqrt{2}$, the number of codewords is $\approx
2.38$ times larger, and a single-user decoder with complexity
${\cal O}(N\log^3 N)$ is provided.

In this paper, we expand on~\cite{TC19}. Based on the underlying
binary symplectic geometry, we provide a systematic way of looking at
the reconstruction algorithm by making use of \emph{stabilizer states}~\cite{DM03} and related notions in quantum computation. This combines the
binary subspace reconstruction discussed in~\cite{TC19} and the BC
reconstruction algorithm of~\cite{HCS08} under the same algebraic framework. Furthermore, we investigate BSSC decoding in true random
access scenarios, where there are multiple randomly selected users
simultaneously accessing the channel. We provide a compressive
sensing multi-user detection algorithm for $L$ simultaneously accessing
randomly selected users with complexity ${\cal O}(N\log^{2+L} N)$. We
find numerically that in a scenario where the channels of the randomly
accessing users come from a continuous complex valued fading
distribution, this multi-BSSC reconstruction algorithm is capable of
reliable multi-user detection.


\section{Preliminaries}
\subsection{The Binary Grassmannian $\cG(m,r;2)$}\label{S-SC}
A binary subspace $H \in \cG(m,r;2)$ is the column space of some matrix $\bH_\cI$ in \emph{column reduced echelon form}, where $\cI \subset \{1,\ldots, m\}$ records the \emph{leading positions}. The dual subspace of $H$ in $\cG(m,m-r;2)$ is the column space of  $\widetilde{\bH_\cI}$, with $(\HI)\T \wHI = 0$. By $\bI_{\cI}$ we will denote the $m\times r$ consisting of the $r$ columns of the identity matrix indexed by $\cI$. Put $\widetilde{\cI}:= \{1,\ldots,m\}\setminus \cI$. Then,
\begin{equation}
\begin{array}{ccc}\label{e-PI1}
(\bI_{\cI})\T\HI = \bI_r, & (\bI_{\cI})\T\bI_{\wcI} = 0, & \wHI\bI_{\wcI} = \bI_{m-r},
\end{array}
\end{equation}
and $\HI$ can be completed to an invertible matrix
\begin{equation}\label{e-PI}
    \PI:=\twomat{\HI}{\bI_{\wcI}}\hspace{.0001 in}\in \GL(m;2).
\end{equation}
The transposed inverse is given by
\begin{equation}\label{e-PIT}
    \PI^{-\sf T} = \twomat{\bI_{\cI}}{\widetilde{\HI}}.
\end{equation}

\subsection{Bruhat Decomposition of the Symplectic Group}
We first briefly describe the symplectic structure of $\Fm$ via the symplectic bilinear form
\begin{equation}\label{e-sinner}
    \inners{\ba,{\bf b}}{\bc,\bd}:= {\bf b}\T\bc+\ba\T\bd.
\end{equation}
A $2m\times 2m$ matrix $\bF$ preserves $\inners{\sbt}{\sbt}$ iff $\bF\bb{$\Omega$}\bF\T = \bb{$\Omega$}$ where 
\begin{equation}
    \Omega = \fourmat{{\bf 0}_m}{\bI_m}{\bI_m}{{\bf 0}_m}.
\end{equation}
We will denote the group of all such \emph{symplectic matrices} $\bF$ with $\Sp(2m;2)$. To proceed, we use the \emph{Bruhat decomposition} of $\Sp(2m;2)$~\cite{Rao93}. For $\bP\in \GL(m;2)$ and $\bS\in \Sym(m;2)$ we distinguish two types of elements in $\Sp(2m;2)$: 
\begin{equation}
    \FDP = \fourmat{\bP}{{\bf 0}_m}{{\bf 0}_m}{\bP^{-\sf T}} \text{ and } \FUS = \fourmat{\bI_m}{\bS}{{\bf 0}_m}{\bI_m}.
\end{equation}
Then every $\bF\in \Sp(2m;2)$ can be written as
\begin{equation}\label{e-Bruhat1}
    \bF = \bF_D(\bP_1)\bF_U(\bS_1)\bF_{\Omega}(r)\bF_U(\bS_2)\bF_D(\bP_2),
\end{equation}
where 
\begin{equation}
    \bF_{\Omega}(r) = \fourmat{\Imrr}{\Imr}{\Imr}{\Imrr},
\end{equation}
with $\Imr$ being the block matrix with $\bI_r$ in upper-left corner and 0 else, and $\Imrr = \bI_m - \Imr$.
We are interested in the right cosets in the quotient group $\Sp(2m;2)/\cP$, where $\cP$ is the subgroup generated by products $\FDP\FUS$. It follows that a coset representative will look like
\begin{equation}\label{e-generic}
    \bF_D(\bP)\bF_U(\bS)\bF_{\Omega}(r),
\end{equation}
for some rank $r$, invertible $\bP$, and symmetric $\bS$. However, two different invertibles $\bP$ may yield representatives of the same coset. We make this precise below.
\begin{lem}[\hspace{-0.001 in}\cite{TC19}]\label{lem}
A right coset in $\Sp(2m;2)/\cP$ is uniquely characterized by a rank $r$, a $m\times m$ symmetric matrix $\widetilde{\bS}_r$ that has $\bS_r\in \Sym(r)$ in its upper-left corner and zero else, and an $r$-dimensional subspace $H$ in $\F_2^m$.
\end{lem}
We will use the coset representative
\begin{equation}\label{e-can}
    \bF_O(\bP_{\cI}, \bS_r):= \bF_D(\bP_{\cI})\bF_U(\widetilde{\bS_r})\bF_{\Omega}(r),
\end{equation}
where $\bP_{\cI}$ as in~\eqref{e-PI} describes $H$.
\subsection{The Heisenberg-Weyl Group}
Fix $N = 2^m$, and let $\{\be_0,\be_1\}$ be the standard basis of $\C^2$. For $\bv\in \F_2^m$ set $\be_{\bv}:=\be_{v_1}\otimes\cdots\otimes \be_{v_m}$. Then $\{\be_{\bv}\mid \bv\in \F_2^m\}$ is the standard basis of $\C^N$. The \emph{Pauli matrices} are
\begin{equation*}
\begin{array}{cccc}
\bI_2, & \!\!\sigma_x = \fourmat{0}{1}{1}{0}\mbox, & \!\!\sigma_z = \fourmat{1}{0}{0}{-1}\mbox, & \!\!\sigma_y = i\sigma_x\sigma_z.
\end{array}
\end{equation*}
For $\ba,{\bf b}\in \F_2^m$ put 
\begin{equation}
    \bD(\ba,{\bf b}):=\sigma_x^{a_1}\sigma_z^{b_1}\otimes\cdots\otimes \sigma_x^{a_m}\sigma_z^{b_m}.
\end{equation}
Directly by definition we have
\begin{equation}\label{e-mult}
    \bD(\ba,{\bf b})\bD(\bc,\bd) = (-1)^{\scriptsize{\bf b}\T\bc}\bD(\ba+\bc,{\bf b}+\bd),
\end{equation}
which in turn implies that $\bD(\ba,{\bf b})$ and $\bD(\bc,\bd)$ commute iff $\inners{\ba,{\bf b}}{\bc,\bd} = 0$.
The \emph{Heisenberg-Weyl} group is defined as
\[
    \cH\cW_N := \{i^k\bD(\ba,{\bf b})\mid \ba,{\bf b}\in \F_2^m, k = 0,1,2,3\} \subset \U(N).
\]
 We will call its elements Pauli matrices as well. 
 Let $\bA$ and ${\bf B}$ be $r\times m$ matrices such $[\bA\,\,\,\,{\bf B}]$ is full rank. We will write
\begin{equation}
  \bE(\bA,\,{\bf B}) : = \{\bE( \bx\T\bA,\, \bx\T{\bf B})\mid  \bx\in \F_2^r\},
\end{equation}
where $\bE(\ba,\,{\bf b}):=i^{\scriptsize \ba\T{\bf b}}\bD(\ba,{\bf b}).$ Here we view the binary vectors as integer vectors and the exponent is taken modulo 4. It follows that
\begin{equation}\label{e-Eab1}
    \bE(\ba,{\bf b}) = i^{\scriptsize \ba\T{\bf b}}\sum_{\scriptsize\bv\in\F_2^m}(-1)^{\scriptsize \bv\T{\bf b}}\be_{\bv+\ba}\be_{\bv}\!\!\!\T.
\end{equation}

Let $\cS = \bE(\bA,\,{\bf B}) \subset \HW$ be a \emph{maximal stabilizer}, that is, a subgroup of $N$ commuting Pauli matrices that does not contain $-\bI_N$, and put 
\begin{equation}\label{e-state}
    V(\cS) := \{\bv\in \C^N\mid \bE\bv=\bv, ~\forall\,\bE \in \cS\}.
\end{equation}
It is well-known (see, e.g.,~\cite{NC00}) that $\dim V(\cS) = 1$. A unit vector that generates it is called \emph{stabilizer state}, and with a slight abuse of notation is also denoted by $V(\cS)$. Because we are disregarding scalars, it is beneficial to think of a stabilizer state as a \emph{Grassmannian line}, that is, $V(\cS)\in \cG(\C^N,1)$.

\section{Clifford Group}
The Clifford group in $N$ dimensions is defined to be the normalizer of $\HW$ in the unitary group $\U(N)$ modulo $\U(1)$: 
\[
    \cl_N=\{\bG\in \U(N)\mid \bG\HW\bG^\dagger = \HW\}/\U(1).
\]
Let $\{\be_1,\ldots,\be_{2m}\}$ be the standard basis of $\Fm$, and consider $\bG\in \cl_N$. Let $\bc_i\in \Fm$ be such that 
\begin{equation}
    \bG\bE(\be_i)\bG^\dagger = \pm\bE(\bc_i).
\end{equation}
Then the matrix $\bF_{\scriptsize\bG}$ whose $i$th row is $\bc_i$ is a symplectic matrix such that
\begin{equation}\label{e-cliff}
    \bG\bE(\bc)\bG^\dagger = \pm\bE(\bc\T\bF_{\scriptsize\bG})
\end{equation}
for all $\bc\in \Fm$. 
We thus have a group homomorphism 
\begin{equation}\label{e-Phi}
    \Phi : \cl_N\longrightarrow \Sp(2m;2),\quad \bG\longmapsto \bF_{\scriptsize\bG},
\end{equation}
with kernel $\ker \Phi = \HW$~\cite{RCKP18}. This map is also surjective; see Section~\ref{S-DC} where specific preimages are given. 
\begin{rem}\label{R-Gd}
Since $\Phi$ is a homomorphism we have that $\Phi(\bG^\dagger) = \bF_{\scriptsize\bG}^{-1}$ and as a consequence $\bG^\dagger\bE(\bc)\bG = \pm\bE(\bc\T\bF_{\scriptsize\bG}^{-1})$.
\end{rem}
\subsection{Decomposition of the Clifford Group}\label{S-DC}
In this section we will make use of the Bruhat decomposition of $\Sp(2m;2)$ to obtain a decomposition of $\cl_N$. To do so we will use the surjectivity of $\Phi$ from~\eqref{e-Phi} and determine preimages of coset representatives from~\eqref{e-can}. The preimages of symplectic matrices $\FDP,\FUS$, and $\bF_{\Omega}(r)$ under $\Phi$ are 
\begin{align}\label{e-GDP}
    \bG_D(\bP) & := \be_\bv\longmapsto \be_{\scriptsize{\bP\T\bv}},\\
    \bG_U(\bS) & := \diag\left(i^{\scriptsize{\bv\T\bS\bv} \mod 4}\right)_{\scriptsize\bv\in\F_2^m}, \\
    \bG_{\Omega}(r) & :=(\bH_2)^{\otimes r}\otimes \bI_{2^{m-r}},
\end{align}
respectively. Here $\bH_2$ is the $2\times 2$ Hadamard matrix. We refer the reader to~\cite[Appendix I]{RCKP18} for details. Directly by the definition of the Hadamard matrix we have
\begin{equation}
    \bH_N:=\bG_{\Omega}(m) = \frac{1}{\sqrt{2^m}}[(-1)^{\scriptsize{\bv\T\bw}}]_{\scriptsize\bv,\bw\in \F_2^m}.
\end{equation}
Whereas, for any $r=1,\ldots,m$, one straightforwardly computes
\begin{equation}\label{e-Gomr}
    \bG_{\Omega}(r)\bZ(m,r) = [(-1)^{\scriptsize{\bv\T \bw}} f(\bv,\bw,r)]_{\scriptsize\bv,\bw\in \F_2^m},
\end{equation}
where $\bZ(m,r) = \bI_{2^r}\otimes \sigma_z^{\otimes m-r}$ are diagonal Pauli matrices, and 
\begin{equation}\label{e-f}
    f(\bv,\bw,r) = \prod_{i=r+1}^m(1+v_i + w_i).
\end{equation}
The value of $f$ will be 1 precisely when $\bv$ and $\bw$ coincide in their last $m-r$ coordinates and 0 otherwise.
\section{Binary Subspace Chirps}
Binary subspace chirps (BSSCs) were introduced in~\cite{TC19} as a generalization of binary chirps (BCs)~\cite{HCS08}. 
In this section we describe the geometric and algebraic features of BSSCs, and use their structure to develop a reconstruction algorithm. 
For each $1\leq r\leq m$, subspace $H\in \cG(m,r;2)$, and symmetric $\bS_r \in \Sym(r;2)$ we will define a unit norm vector in $\C^N$ as follows. 
Let $H$ be the column space of $\HI$, as described in Section~\ref{S-SC}. 
Then $\HI$ is completed to an invertible $\bP:=\PI$ as in~\eqref{e-PI}. For all ${\bf b}, \ba \in \F_2^m$ define
\[
     \bw_{\scriptsize{\bf b}}^{\scriptsize{H,\bS_r}}(\ba) = \frac{1}{\sqrt{2^r}}i^{\scriptsize{\ba\T\bP^{-\sf T}\bS\bP^{-1}\ba+2{\bf b}\T\bP^{-1}\ba}}f({\bf b},\bP^{-1}\ba,r),
\]
where $\bS\in \Sym(m;2)$ is the matrix with $\bS_r$ on the upper-left corner and 0 elsewhere, $f$ is as in~\eqref{e-f}, and the arithmetic in the exponent is done modulo 4. To avoid heavy notation however we will omit the upper scripts. Then we define a \emph{binary subspace chirp} to be 
\begin{equation}
 \bw_{\scriptsize{\bf b}} := [ \bw_{\scriptsize{\bf b}}(\ba)]_{\scriptsize\ba\in \F_2^m}\in \C^N.
\end{equation}
Note that when $r=m$ we have $\bP = \bI_m$ and $f$ is the identically 1 function. Thus, one obtains the \emph{binary chirps} \cite{HCS08} as a special case.

Directly from the definition (and the definition of $f$) it follows that $ \bw_{\scriptsize{\bf b}}(\ba)\neq 0$ precisely when ${\bf b}$ and $\bP^{-1}\ba$ coincide in their last $m-r$ coordinates. Making use of the structure of $\bP$ as in \eqref{e-PI} we may conclude that $ \bw_{\scriptsize{\bf b}}(\ba) \neq 0$ iff
\begin{equation}\label{e-bmr}
    \wHI\vspace{-.004 in}^{\T}\ba = {\bf b}_{m-r},
\end{equation}
where ${\bf b}_{m-r}\in \F_2^{m-r}$ consists of the last $m-r$ coordinates of ${\bf b}$. It follows that $ \bw_{\scriptsize{\bf b}}$ has $2^r$ non-zero entries, and thus it is a unit norm vector. Making use of~\eqref{e-PI1} we see that the solution space of~\eqref{e-bmr} is given by 
\begin{equation}\label{e-bmr1}
    \{\widetilde{ \bx}:=\bI_{\widetilde{\cI}}{\bf b}_{m-r} + \HI \bx\mid  \bx\in \F_2^r\}.
\end{equation}
We say that $\HI$ determines the \emph{on-off pattern} of $ \bw_{\scriptsize{\bf b}}$.

\begin{rem}\label{R-emb}
Fix a subspace chirp $\wb$, and write ${\bf b}\T = [{\bf b}\T_{\!\!\!\!r}\,\,\,\, {\bf b}\T_{\!\!\!\!m-r}]$. Then $\wb(\ba) \neq 0$ iff $\ba$ is as in \eqref{e-bmr1} for some $ \bx\in \F_2^r$. Making use of \eqref{e-PIT} and \eqref{e-PI1} we obtain
\begin{equation}\label{e-p1a1}
    \bP^{-1}\ba = \left[ \!\!\!\begin{array}{c} \bx \\ {\bf b}_{m-r}\end{array}\!\!\!\right],
\end{equation}
and as a consequence $\ba\T\bP^{-\sf T}\bS\bP^{-1}\ba =  \bx\T\bS_r \bx$ where $\bS_r$ is the (symmetric) upper-left $r\times r$ block of $\bS$. Thus the nonzero entries of $\wb$ are of the form
\begin{equation}\label{e-p1a}
    \wb( \bx) = \frac{(-1)^{\wt(\scriptsize{{\bf b}_{m-r})}}}{\sqrt{2^r}}i^{\scriptsize{ \bx\T\bS_r \bx} + 2{\bf b}\T_{\!\!\!r} \bx}
\end{equation}
for $ \bx\in\F_2^r$. Note that there is a slight abuse of notation where we have identified $ \bx$ with $\bP^{-1}\ba$ (thanks to \eqref{e-p1a1} and the fact that ${\bf b}$ is fixed). Above, the function $\wt(\sbt)$ is just the Hamming weight which counts the number of non-zero entries in a binary vector. We conclude that the \emph{on-pattern} of a rank $r$ binary subspace chirp is just a binary chirp in $2^r$ dimensions; compare \eqref{e-p1a} with \cite[Eq.~(5)]{HCS08}. It follows that all lower-rank chirps are embedded in $2^m$ dimensions, which along with all the chirps in $2^m$ dimensions yield all the binary subspace chirps. As discussed, the embeddings are determined by subspaces.
\end{rem}
\subsection{Algebraic Structure of BSSCs}
Let $\bG_{\bF} = \bG_D(\bP\T) \bG_U(\bS)\bG_{\Omega}(r)$, that is, $\Phi(\bG_{\scriptsize \bF}) = \bF$. Recall also that $\{\be_{\ba}\mid \ba\in \F_2^m\}$ is the standard basis of $\C^N$. If we put $\bu:= \bP^{-1}\ba$ we have 
\begin{align}
     \bw_{\scriptsize{\bf b}} & = \frac{1}{\sqrt{2^r}}\sum_{\scriptsize\ba\in \F_2^m}  \bw_{\scriptsize{\bf b}}(\ba)\be_{\scriptsize\ba} \nonumber \\
    & = \frac{1}{\sqrt{2^r}}\sum_{\scriptsize\bu\in \F_2^m}i^{\scriptsize{\bu\T\bS\bu}}(-1)^{\scriptsize{{\bf b}\T\bu}}f({\bf b},\bu,r)\be_{\scriptsize{\bb{Pu}}} \nonumber\\
     & = \bG_D(\bP\T)\cdot\frac{1}{\sqrt{2^r}} \sum_{\scriptsize\bu\in \F_2^m}i^{\scriptsize{\bu\T\bS\bu}}(-1)^{\scriptsize{{\bf b}\T\bu}}f({\bf b},\bu,r)\be_{\scriptsize{\bu}}\nonumber \\
    & = \bG_D(\bP\T)\bG_U(\bS)\bG_{\Omega}(r)\bZ(m,r)\be_{\scriptsize{{\bf b}}} \label{e-eqq1}\\
    & = \bG_{\scriptsize{\bF}}\cdot\bZ(m,r)\be_{\scriptsize{{\bf b}}}\label{e-eqq2},
\end{align}
where~\eqref{e-eqq1} follows by~\eqref{e-Gomr}. Note that in~\eqref{e-eqq2}, the diagonal Pauli $\bZ(m,r)$ only ever introduces an additional sign on columns of $\bG_{\scriptsize{\bF}}$. Thus, the binary subspace chirp $ \bw_{\scriptsize{{\bf b}}}$ is nothing else but the ${\bf b}$th column of $\bG_{\scriptsize{\bF}}$, up to a sign. However, as mentioned, for our practical purposes a sign (or even a complex unit) is irrelevant.

Since commuting matrices can be simultaneously diagonalized, it is natural to consider the common eigenspace of maximal stabilizers. We have the following.
\begin{theo}\label{T-ES}
Let $\bF$ and $\bG_{\scriptsize\bF}$ be as above. The set $\{ \bw_{\scriptsize{\bf b}}~\mid{\bf b}\in \F_2^m\}$, that is the columns of $\bG_{\scriptsize\bF}$, is the common eigenspace of the maximal stabilizer $\bE(\Imr\bP\T,\,(\Imr\bS+\Imrr)\bP^{-1})$.
\end{theo}
\begin{proof}
Consider the matrix $\bG:=\bG_{\scriptsize{\bF}}$ parametrized by the symplectic matrix $\bF$, and recall that $\wb$ is the ${\bf b}$th column of $\gf$. It follows from Remark~\ref{R-Gd} that the columns of $\bG$ are the eigenspace of $\bE( \bx,\by)$ iff 
\begin{equation}
    \bG^\dagger \bE( \bx,\by)\bG = \pm \bE([ \bx,\by]\T\bF^{-1})
\end{equation}
is diagonal. Recall also that $\bE( \bx,\by)$ is diagonal iff $ \bx = {\bf 0}$, and observe that $\bF_{\Omega}(r)^{-1} = \bF_{\Omega}(r)$. Thus, $\bG_{\Omega}(r)$ will be the common eigenspace of the maximal stabilizer $\cS$ iff $\pm\bE([ \bx\,\,\by]^{\T}\bF_{\Omega}(r))$ is diagonal for all $\bE( \bx,\by)\in \cS$. Then it is easy to see that such a maximal stabilizer is $\bE(\Imr,\,\Imrr)$. Next, if $ \bw$ is an eigenvector of $\bE(\bc)$ then
\begin{align*}
   \bG\bw = \pm\bG\bE(\bc) \bw & = \pm\bG\bE(\bc)\bG^\dagger\bG \bw 
   = \pm\bE(\bc\T\Phi(\bG))\bG\bw
   \end{align*}
implies that $\bG\bw$ is an eigenvector of $\bE(\bc\T\Phi(\bG))$. The proof is concluded by computing $[\Imr\,\,\,\,\Imrr]\bF_U(\bS)\bF_D(\bP\T)$.
\end{proof}
\begin{rem}
For $r=m$ one has $\bE(\Imr,\,\Imrr) = \bE(\bI_m,\, {\bf 0})$ and $\bG_{\Omega}(r) = \bH_N$. 
Thus the above theorem covers the well-known fact that $\bH_N$ is the common eigenspace of $\bE(\bI_m,\, {\bf 0})$. In this extremal case we also have $\bP_{\cI} = \bI_m$ and $\widetilde{\bS_r} = \bS\in \Sym(m;2)$. So the above theorem also covers~\cite[Lem.~11]{CRCP19} which (in the language of this paper) says  that the common eigenspace of $\bE(\bI_m,\,\bS)$ is $\bG_U(\bS)\bH_N$.
\end{rem}
\subsection{Reconstruction of Single BSSC}
Now we shall use the underlying algebraic structure of BSSCs summarized in Theorem~\ref{T-ES} to determine a reconstruction algorithm that unifies the identification of the binary subspace $H$~\cite{TC19}, and the symmetric matrix $\bS$~\cite{HCS08}.
We focus first on noise-free reconstruction. The easiest task is the recovery of the rank $r$. Namely, by \eqref{e-bmr1} we have
\begin{equation}
    \wb(\ba)\ov{\wb(\ba)} = \left\{\!\!\begin{array}{ll} 1 /2^r, & 2^r \text{ times,} \\ 0, & 2^{m-r}\text{ times.}\end{array} \right.
\end{equation}
To reconstruct $\bS_r$ and then eventually $H$ we modify the \emph{shift and multiply} technique used in \cite{HCS08} for the reconstruction of binary chirps. However, in our scenario extra care is required as the shifting can perturb the on-off pattern. Namely, we must use only shifts $\ba\longmapsto \ba+\be$ that preserve the on-off pattern. It follows by \eqref{e-bmr} that we must use only shifts by $\be$ that satisfy $\widetilde{\bH_\cI}\T\be = {\bf 0}$, or equivalently $\be = \bH_\cI\by$ for $\by\in \F_2^r$. In this instance, thanks to \eqref{e-PI1} we have 
\begin{equation}
    \bP^{-1}\be = \bP^{-1}\bH_\cI\by = \twomatv{\by}{{\bf 0}}.
\end{equation}
If we focus on the nonzero entries of $\wb$ and on shifts that preserve the on-off pattern of $\wb$ we can make use of Remark~\ref{R-emb}, where with another slight abuse of notation we identify $\by$ with $\bP^{-1}\be$. It is beneficial to take $\by$ to be ${\bf f}_i$ - one of the standard basis vectors of $\F_2^r$. With this preparation we are able to use the shift and multiply technique:
\begin{equation}\label{e-shift}
    \wb(\bx+{\bf f}_i)\ov{\wb(\bx)} = \frac{1}{2^r}\cdot i^{\scriptsize {\bf f}\T_{\!\!i}\bS_r{\bf f}_i}\cdot(-1)^{\scriptsize {\bf b}\T_{\!\!r}{\bf f}_i}\cdot(-1)^{\scriptsize \bx\T\bS_r{\bf f}_i}.
\end{equation}
Note that above only the last term depends on $ \bx$. Next, multiply \eqref{e-shift} with the Hadamard matrix to obtain
\begin{equation}\label{e-shift1}
  i^{\scriptsize {\bf f}\T_{\!\!i}\bS_r{\bf f}_i}\cdot(-1)^{\scriptsize {\bf b}\T_{\!\!r}{\bf f}_i}\sum_{\scriptsize \bx\in \F_2^r}(-1)^{\scriptsize \bx\T(\bv+\bS_r{\bf f}_i)},
\end{equation}
for all $\bv\in \F_2^r$ (where we have omitted the scaling factor). Then \eqref{e-shift1} is nonzero precisely when $\bv = \bS_r{\bf f}_{i}$ - the $i$th column of $\bS_r$. With $\bS_r$ in hand, one recovers ${\bf b}_r$ similarly by multiplying $\wb( \bx)\ov{ \bw_{\scriptsize {\bf 0}}( \bx)}$ with the Hadamard matrix. To recover ${\bf b}_{m-r}$ one simply uses the knowledge of nonzero coordinates and \eqref{e-p1a1}. Next, with ${\bf b}$ in hand and the knowledge of the on-off pattern one recovers $\bH_\cI$ (and thus $H$) using \eqref{e-bmr} or equivalently \eqref{e-bmr1}.

In the above somewhat ad-hoc method we did not take advantage of the geometric structure of the subspace chirps as eigenvectors of given maximal stabilizers or equivalently as the columns of given Clifford matrices. We do this next by following the line of \cite{TC19}.

Let $ \bw$ be a subspace chirp, and recall that it is a column of $\bG:=\bG_{\scriptsize{\bF}} = \bG_D(\bP\T) \bG_U(\bS)\bG_{\Omega}(r)$ where $\bF:=\bF_{\Omega}(r)\FUS\bF_D(\bP\T)$. Then by construction $\bG$ and $\bF$ satisfy $\bG^\dagger\bE(\bc)\bG = \pm\bE(\bc\T\bF^{-1})$ for all $\bc\in \F_2^{2m}$. Recall also from Theorem~\ref{T-ES} that $\bG$ is the common eigenspace of the maximal stabilizer 
\begin{equation}\label{e-ES1}
    \bE(\Imr\bP\T,\,(\Imr\bS+\Imrr)\bP^{-1}\!) = \bE\!\left(\!\left[\!\begin{array}{c|c} \!\bH_{\cI}\!\!\!^{\T}\!&\!\bS_r\bI_{\cI}^{\T}\!\! \\ \!{\bf 0}&\!\widetilde{\bH_{\cI}}^{\T}\!\!\end{array} \!\right] \!\right).
\end{equation}
Thus, to reconstruct the unknown subspace chirp $ \bw$ it is sufficient to first identify the maximal stabilizer that stabilizes it, and then identify $ \bw$ as a column of $\bG$. A crucial observation at this stage is the fact that the maximal stabilizer in \eqref{e-ES1} has precisely $2^r$ off-diagonal and $2^{m-r}$ diagonal Pauli matrices.

We now make use of the argument in Theorem~\ref{T-ES}, that is, $ \bw$ is an eigenvector of $\bE(\bc)$ iff $\bE(\bc\T\bF^{-1})$ is diagonal. 
Let us focus first on identifying the diagonal Pauli matrices that stabilize $ \bw$, that is, $\bc\T = [{\bf 0}\,\,\,\,\by\T]$. 
Then for such $\bc$, $ \bw$ is an eigenvector of $\bE(\bc)$ iff $\by\T\bH_{\cI} = 0$ iff $\by = \widetilde{\bH_{\cI}}\bz$ for some $\bz\in\F_2^{m-r}$. 
Thus, to identify the diagonal Pauli matrices that stabilize $ \bw$, and consequently the subspaces $\bH_{\cI}$ and $\widetilde{\bH_{\cI}}$, it is sufficient to find $2^{m-r}$ vectors $\by\in \F_2^{m}$ such that $ \bw^\dagger \bE({\bf 0},\by) \bw \neq 0$.
It follows by \eqref{e-Eab1} that the latter is equivalent with finding $2^{m-r}$ vectors $\by$ such that 
\begin{equation}
    \sum_{\scriptsize \bv\in\F_2^m}(-1)^{\scriptsize \bv\T\by}| \bw(\bv)|^2 \neq 0.
\end{equation}
The above is just a Hadamard transform which can be efficiently undone. With a similar argument, $ \bw$ is an eigenvector of a general Pauli matrix $\bE( \bx,\by)$ iff 
\begin{equation}\label{e-shift2}
  \hspace{-0.05 in} \bw^\dagger\bE( \bx,\by) \bw = i^{\scriptsize  \bx\T\by}\!\!\sum_{\scriptsize \bv\in\F_2^m}(-1)^{\scriptsize \bv\T\by}\ov{ \bw(\bv\!+\! \bx)} \bw(\bv) \neq 0.
\end{equation}
This is again just a Hadamard transform.

Let us now explicitly make use of \eqref{e-shift2} to reconstruct the symmetric matrix $\bS_r$, while assuming that we have already reconstructed $\bH_{\cI}, \widetilde{\bH_{\cI}}$. We first have
\begin{equation}
    \bF^{-1} = \left[\!\!\begin{array}{cccc} \bI_{\cI}\bS_r & \widetilde{\bH_{\cI}} & \bI_{\cI} & {\bf 0} \\ \bH_{\cI} & {\bf 0}&{\bf 0}&\bI_{\widetilde{\cI}}\end{array}\!\!\right].
\end{equation}
Then, for $\bc = \left[\!\!\begin{array}{c}  \bx\\\by \end{array}\!\!\right]$, we have $ \bw^\dagger\bE( \bx,\by) \bw\neq 0$ iff $\bE(\bc\T\bF^{-1})$ is diagonal, iff
\begin{equation}\label{e-SrRec}
     \bx\T[\bI_{\cI}\bS_r \,\,\,\, \widetilde{\bH_{\cI}}] = \by\T[\bH_{\cI}\,\,\,\,{\bf 0}].
\end{equation}
We are interested in $\by\in\F_2^m$ that satisfy \eqref{e-SrRec}. First note that solutions to \eqref{e-SrRec} exist only if $ \bx\T\widetilde{\bH_{\cI}} = {\bf 0}$, 
i.e.,
if $ \bx = \bH_{\cI}\bz$, $\bz\in \F_2^r$. For such $ \bx$, making use of \eqref{e-PI1}, we conclude that \eqref{e-SrRec} holds iff $\bz\T\bS_r = \by\T\bH_{\cI}$, solutions of which are given by
\begin{equation}\label{e-SrRec2}
    \by = \widetilde{\bH_{\cI}}\bv + \bI_{\cI}\bS_r\bz,\,\,\,\,\bv\in\F_2^{m-r}.
\end{equation}
If we take $\bz = {\bf f}_{i}$ - the $i$th standard basis vector of $\F_2^r$ - we have that $\bz\T\bS_r$ is the $i$th row/column of $\bS_r$ while $ \bx = \bH_{\cI}\bz$ is the $i$th column of $\bH_{\cI}$.

We resume everything to the following algorithm.


\begin{algorithm}\caption{Reconstruction of single noiseless BSSC}\label{alg}
{\bf Input:} Unknown BSSC $\bw$
\begin{algorithmic}
\STATE~1. Compute $\bw^\dagger \bE({\bf 0},\by)\bw$ for $\by\in \F_2^m$.
\STATE~2. Find $\bH_{\cI}$ using 
\[
    \bw^\dagger \bE({\bf 0},\by)\bw \neq 0 \text{ iff }\by\T\bH_{\cI} = {\bf 0} \text{ iff } \by\in \cs(\widetilde{\bH_{\cI}}).
\]
\vspace{-0.2 in}
\STATE~3. Construct $\PI$ as in~\eqref{e-PI}.
\STATE~4. $r = \rk(\bH_{\cI})$. 
\STATE~5. $\bb{for}$ $i = 1,\ldots,r$ do:
\STATE~6. \quad Compute $\bw^\dagger \bE(\bH_{\cI}{\bf f}_i,\by)\bw$ for $\by\in\F_2^m$.
\STATE~7. \quad Determine the $i$th row of $\bS_r$ using \eqref{e-SrRec2}.
\STATE~8. {\bf end for}
\STATE~9. Dechirp $\bw$ to find ${\bf b}$.
\end{algorithmic}
{\bf Output:} $r,\bS_r,\PI,{\bf b}$.
\end{algorithm}


\section{Multi-BSSC Reconstruction}

In~\cite{TC19} a reconstruction algorithm of a single BSSC in the presence of 
noise was presented. The algorithm makes $m+1$ rank hypothesis, and for each hypothesis the on-off pattern is estimated. The best BSSC among the $m+1$ is output. A similar strategy can be used to generalize Algorithm~\ref{alg} to decode multiple simultaneous transmissions in a multi-user scenario 
\begin{equation}\label{e-linBSSC}
    \bs = \sum_{\ell = 1}^L h_\ell \bw_\ell.
\end{equation}
Here the channel coefficients $h_\ell$ are complex valued, and can be modeled as $\cC\cN(0,1)$, and $ \bw_\ell$ are BSSCs. This represents, e.g., a random access scenario, where $L$ randomly chosen active users transmit a signature sequence, and the receiver should identify the active users. 

We generalize the single-user 
algorithm to a multi-user algorithm, where the coefficients $h_\ell$ are estimated in the process of identifying the most probable transmitted signals. For this, we use Orthogonal Matching Pursuit (OMP), which is analogous with the strategy of~\cite{HCS08}. We assume that we know the number of active users $L$. 
\begin{algorithm}\caption{Reconstruction of noiseless multi-BSSCs}\label{alg1}
{\bf Input:} Signal~$\bs$ as in~\eqref{e-linBSSC}.
\begin{algorithmic}
\STATE~1. \textbf{for} $\ell=1:L$ do
\STATE~2. \quad \textbf{for} $r = 0:m$ do
\STATE~3. \qquad Greedily construct the $m-r$ dimensional subspace \\ 
\hspace{.425 in} $\wHI$ using the highest values of $|\bs^\dagger\bE({\bf 0},\by)\bs|$.
\STATE~4. \qquad Estimate $\widetilde{\bw}_r$ as in Alg.~\ref{alg}.
\STATE~5. \quad \textbf{end for} 
\STATE~6. \quad Select the best estimate $\widetilde{\bw}_\ell$.
\STATE~7. \quad Determine $\widetilde{h}_1,\ldots,\widetilde{h}_\ell$ that minimize
$$\left\|\bs - \sum_{j = 1}^\ell h_j\widetilde{\bw}_j\right\|_2$$.
\STATE~8. \quad Reduce $\bs$ to $\bs' = \bs - 
 \sum_{j = 1}^\ell \widetilde{h}_j\widetilde{\bw}_j$.
 \STATE~9. \textbf{end for}
\end{algorithmic}
\textbf{Output: }$\widetilde{\bw}_1,\ldots, \widetilde{\bw}_L$.
\end{algorithm}
The estimated error probability of single user transmission for $L = 2,3$ is given in Figure~\ref{fig}. For the simulation, the BSSCs are chosen uniformly at random from the codebook.
We compare the results with BC codebooks and random codebooks with the same cardinality. For random codebooks, steps (2)-(5) are substituted with exhaustive search (which is infeasible beyond $m = 6$).

The erroneous reconstructions of Algorithm~\ref{alg1} come in part from steps (3)-(4). Specifically, from the cross-terms of 
\begin{equation*}
    \bs^\dagger\bs = \sum_{\ell = 1}^L|h_\ell|^2\| \bw_\ell\|^2 + \sum_{i\neq \ell} \ov{h_i}h_\ell \bw_i^\dagger \bw_\ell.
\end{equation*}
For BCs, these cross-terms are the well-behaved \emph{second order Reed-Muller functions}. On the other hand, the BSSCs, unlike the BCs~\cite{CHJ10}, do not form a group under point-wise multiplication, and thus the products $ \bw_i^\dagger \bw_\ell$ are more complicated. In addition, linear combinations of BSSCs~\eqref{e-linBSSC} may perturb each others on-off pattern and depending on the nature of the channel coefficients $h_\ell$, the algorithm may detect a higher rank BSSC in $\bs$. 
If the channel coefficients of two low rank BSSCs happen to have similar amplitudes, the algorithm may detect a lower rank BSSC that corresponds to the overlap of the on-off patterns of the BSSCs. Despite these scenarios, an elaborate decoding algorithm like the one discussed, is able to provide reliable performance.

\begin{figure}
\centering
\includegraphics[width = 0.5\textwidth]{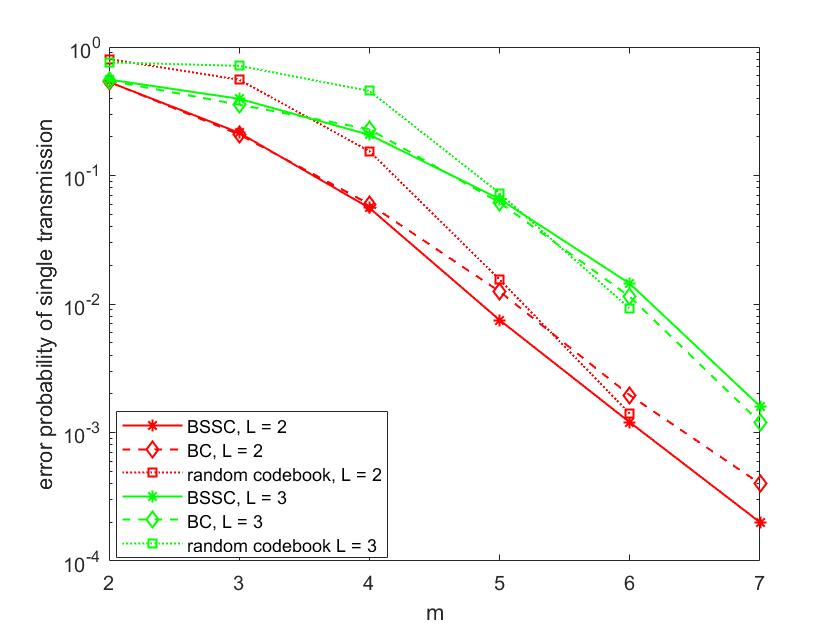}
\caption{Error probability of Algorithm~\ref{alg1}.} 
\label{fig}
\end{figure}

Interestingly, BSSCs outperform BCs, despite these codebooks having the same minimum distance. In~\cite{TC19}, the same was observed in single-user reconstruction. With increasing $m$, the performance benefit of the algebraically defined codebook over random codebooks diminishes. However, the decoding complexity remains manageable for the algebraic codebooks.

\section{Conclusion}
We have extended the work~\cite{TC19} by exploiting the geometry of BSSCs. These Grassmannian lines are described as common eigenspaces of maximal sets of commuting Pauli matrices, or equivalently, as columns of Clifford matrices. Further, we have developed a low complexity algorithm for multi BSSCs transmission with low error probability. In future research, we shall consider also noise in
multi-user reconstruction, and work toward a practical algorithm along the lines of~\cite{Calderbank2019}.


\section*{Acknowledgements}
This work was funded in part by the Academy of Finland (grants 299916, 319484).

\bibliographystyle{IEEEtran}
\bibliography{IEEEabrv,BSSC}

\end{document}